\documentclass[12pt,a4paper]{article}

\usepackage{amsmath,amssymb}
\usepackage{mathrsfs}
\usepackage{mathtools}
\usepackage{graphicx}
\usepackage[nosort]{cite}

\usepackage{tikz}
\usetikzlibrary{decorations.markings}
\usepackage{epsfig}  
\usepackage{enumerate}
\usepackage{relsize}
\usepackage{dsfont}
\newcommand{\mathsym}[1]{{}}

\usepackage{amsthm,amsbsy,color,multicol}
\usepackage{array,calc}
\usepackage{rotating}
\usepackage{a4wide,bm}
\usepackage{bbm}
\usepackage[a4paper,text={450pt,650pt},centering]{geometry}
\usepackage{setspace}

\let\oldbfseries=\bfseries
\let\oldmdseries=\mdseries
\let\oldnormalfont=\normalfont
\renewcommand{\bfseries}{\oldbfseries\boldmath}
\renewcommand{\mdseries}{\oldmdseries\unboldmath}
\renewcommand{\normalfont}{\oldnormalfont\unboldmath}

\allowdisplaybreaks[3]

\numberwithin{equation}{section}

\usepackage[font=small,labelfont=bf,width=0.85\textwidth]{caption}

\ifx\hypersetup\sadfkjashdfkxja\newcommand\hypersetup[1]{}\fi

\hypersetup{plainpages=false}
\hypersetup{pdfpagemode=UseNone}
\hypersetup{bookmarksnumbered=true}
\hypersetup{pdfstartview=FitH}
\hypersetup{colorlinks=false}
\hypersetup{citebordercolor={.5 1 .5}}
\hypersetup{urlbordercolor={.5 1 1}}
\hypersetup{linkbordercolor={1 .7 .7}}

\DeclareMathSymbol{\Gamma}{\mathalpha}{letters}{"00}
\DeclareMathSymbol{\Delta}{\mathalpha}{letters}{"01}
\DeclareMathSymbol{\Theta}{\mathalpha}{letters}{"02}
\DeclareMathSymbol{\Lambda}{\mathalpha}{letters}{"03}
\DeclareMathSymbol{\Xi}{\mathalpha}{letters}{"04}
\DeclareMathSymbol{\Pi}{\mathalpha}{letters}{"05}
\DeclareMathSymbol{\Sigma}{\mathalpha}{letters}{"06}
\DeclareMathSymbol{\Upsilon}{\mathalpha}{letters}{"07}
\DeclareMathSymbol{\Phi}{\mathalpha}{letters}{"08}
\DeclareMathSymbol{\Psi}{\mathalpha}{letters}{"09}
\DeclareMathSymbol{\Omega}{\mathalpha}{letters}{"0A}

\newcommand{\gen}[1]{\mathrm{#1}}

\newcommand{\ii}{\mathrm{i}}

\newcommand*\widebar[1]{%
  \hbox{%
    \vbox{%
      \hrule height 0.5pt 
      \kern0.25ex
      \hbox{%
        \kern-0.3em
        \ensuremath{#1}%
        \kern-0.1em
      }%
    }%
  }%
}

\ifx\genfrac\sdflkaj\else\fi

\newcommand{\ket}[1]{\left|#1\right\rangle}      
\newcommand{\bra}[1]{\left\langle #1\right|}     

\newcommand{\alg}[1]{\mathfrak{#1}}

\newcommand{\beq}{\begin{equation}}
\newcommand{\eeq}{\end{equation}}

\def\[{\begin{equation}}
\def\]{\end{equation}}
\def\<{\begin{eqnarray}}
\def\>{\end{eqnarray}}

\newtheorem{mydef}{Definition}

\newtheorem{theorem}{Theorem}
 
\newtheorem{remark}{Remark}

\newtheorem{corollary}{Corollary}
\newtheorem{example}{Example}

\makeatletter
\def\mr@ignsp#1 {\ifx\:#1\@empty\else #1\expandafter\mr@ignsp\fi}%
\newcommand{\multiref}[1]{\begingroup
\xdef\mr@no@sparg{\expandafter\mr@ignsp#1 \: }%
\def\mr@comma{}%
\@for\mr@refs:=\mr@no@sparg\do{\mr@comma\def\mr@comma{,}\ref{\mr@refs}}%
\endgroup}
\makeatother

\newcommand{\hypref}[2]{\ifx\href\asklfhas #2\else\href{#1}{#2}\fi}
\newcommand{\Secref}[1]{Section~\multiref{#1}}

\newcommand{\Appref}[1]{Appendix~\multiref{#1}}
\newcommand{\appref}[1]{App.~\multiref{#1}}

\renewcommand{\eqref}[1]{(\multiref{#1})}

\makeatletter
\newlength{\apb@width}
\newcommand{\autoparbox}[2][c]{\settowidth{\apb@width}{#2}\parbox[#1]{\apb@width}{#2}}

\makeatother

\ifx\href\asklfhas\newcommand{\href}[2]{#2}\fi

\begin{document}

\renewcommand{\thefootnote}{\fnsymbol{footnote}}
\thispagestyle{empty}
\begin{flushright}\footnotesize
ITP-UU-13/32 \\
SPIN-13/24
\end{flushright}
\vspace{1cm}

\begin{center}%
{\Large\bfseries%
\hypersetup{pdftitle={Twisted Heisenberg chain and the six-vertex model with DWBC}}%
Twisted Heisenberg chain and \\ the six-vertex model with DWBC%
\par} \vspace{2cm}%

\textsc{W. Galleas}\vspace{5mm}%
\hypersetup{pdfauthor={Wellington Galleas}}%

\textit{Institute for Theoretical Physics and Spinoza Institute, \\ Utrecht University, Leuvenlaan 4,
3584 CE Utrecht, \\ The Netherlands}\vspace{3mm}%

\verb+w.galleas@uu.nl+ %

\par\vspace{3cm}

\textbf{Abstract}\vspace{7mm}

\begin{minipage}{12.7cm}
In this work we establish a relation between the six-vertex model with 
Domain Wall Boundary Conditions (DWBC) and the $XXZ$ spin chain with anti-periodic
twisted boundaries. More precisely, we demonstrate a formal relation between the zeroes
of the partition function of the six-vertex model with DWBC and the zeroes of the transfer matrix
eigenvalues associated with the six-vertex model with a particular non-diagonal boundary twist.

\hypersetup{pdfkeywords={Heisenberg chain, domain wall boundaries, functional equations}}%
\hypersetup{pdfsubject={}}%
\end{minipage}
\vskip 2cm
{\small PACS numbers:  05.50+q, 02.30.IK}
\vskip 0.1cm
{\small Keywords: Heisenberg chain, domain wall boundaries, functional equations}
\vskip 2cm
{\small October 2014}

\end{center}

\newpage
\renewcommand{\thefootnote}{\arabic{footnote}}
\setcounter{footnote}{0}

\tableofcontents

\section{Introduction}
\label{sec:intro}

One of the remarkable roles played by integrable systems is the establishment of connections
between seemingly unrelated topics. For instance, although the relation between one-dimensional
quantum spin chains and two-dimensional classical vertex models is nowadays clear, this remarkable relation has its origins in 
Lieb's observation that the ice model transfer matrix and the $XXX$ spin chain hamiltonian
share the same eigenvectors \cite{Lieb_1967}.
In fact, this relation was only made clear by Baxter \cite{Baxter_1971} whom showed that the logarithmic derivative
of a two-dimensional vertex model transfer matrix gives rise to an one-dimensional quantum spin chain hamiltonian.
This correspondence between quantum spin chains and classical vertex models is well
established for lattice systems but we also have further connections emerging in the continuum limit.
For instance, it is believed that the massless regimes of vertex models in the continuum are described by the critical
properties of Wess-Zumino-Witten field theories \cite{Knizhnik_1984}. 

As far as vertex models with Domain Wall Boundary Conditions (DWBC) \cite{Korepin_1982} are concerned,
we can not immediately associate an one-dimensional spin chain along the lines of \cite{Baxter_1970}.
Nevertheless, the six-vertex model with domain wall boundaries still exhibits interesting relations with
the theory of classical integrable systems \cite{Foda_2009}, special functions \cite{Lascoux_2007, Rosengren_2011}
and enumerative combinatorics \cite{Kuperberg_1995}. Moreover, in the recent paper \cite{Galleas_proc} we have shown that the
partition function of the six-vertex model with DWBC corresponds to the null eigenvalue 
wave-function of a certain many-body hamiltonian operator.

On the other hand, the $XXZ$ spin chain with anti-periodic boundary conditions can be embedded in the transfer matrix
of a $\mathcal{U}_q [\widehat{\alg{sl}}(2)]$ invariant six-vertex model with a particular non-diagonal boundary twist
along the same lines of \cite{Baxter_1970, deVega_1984}.
This particular spin chain has also been studied in \cite{Baxter_1995, Galleas_2008, Niccoli_2013, WengLi_2013}
and it was the first system tackled through the algebraic-functional method used in \cite{Galleas_2010, Galleas_2011, Galleas_2012} 
for partition functions with domain wall boundaries. This method has been refined in a series of papers and here we intend to report
on a novel connection between the twisted Heisenberg chain and the six-vertex model with DWBC. 

This paper is organized as follows. In \Secref{sec:heisenberg} we briefly describe the transfer
matrix embedding the $XXZ$ chain with anti-periodic boundary conditions and introduce the notation
we shall use throughout this paper. In \Secref{sec:FN} we explore the Yang-Baxter algebra along the lines
of \cite{Galleas_proc} in order to derive a functional equation relating the transfer matrix eigenvalues and the
partition function of the six-vertex model with DWBC. The consequences of this functional equation
is discussed in \Secref{sec:SOL} and, in particular, we show how our results can be simplified
when the anisotropy parameter is a root of unity. Concluding remarks are then discussed in \Secref{sec:conclusion}
and the \Appref{sec:proof} is devoted to the derivation of our main result.

\section{Heisenberg chain and the DWBC partition function}
\label{sec:heisenberg}

In this section we shall give a brief description of the anisotropic Heisenberg chain
with anti-periodic boundary conditions. This model consists of a spin-$\frac{1}{2}$
system and here we shall mostly adopt the conventions of \cite{Galleas_2008}.
The system hamiltonian $\mathcal{H}$ acts on the tensor product space
$\mathbb{V}_{\mathcal{Q}} \cong (\mathbb{C}^2)^{\otimes L}$ and it reads
\[
\label{ham}
\mathcal{H} \coloneqq \sum_{i=1}^L \left( \sigma^x_{i} \sigma^x_{i+1} + \sigma^y_{i} \sigma^y_{i+1} + \cosh{(\gamma)} \sigma^z_{i} \sigma^z_{i+1} \right) \qquad \in \quad \mbox{End}(\mathbb{V}_{\mathcal{Q}}) \; .
\]
In (\ref{ham}) we have employed the notation $\sigma_i = \mathbbm{1}^{\otimes (i-1)} \otimes \sigma \otimes \mathbbm{1}^{\otimes (L-i)}$
where $\sigma \in \{\sigma^x , \sigma^y , \sigma^z \}$ denotes the standard Pauli matrices and $\mathbbm{1}$ stands for the 
identity matrix in $\mbox{End}( \mathbb{C}^2 )$.
As far as boundary terms are concerned, here we consider the following anti-periodic conditions:
\[
\sigma^{x}_{L+1} \coloneqq \sigma^{x}_{1} \qquad \quad \sigma^{y}_{L+1} \coloneqq - \sigma^{y}_{1} \qquad \quad \sigma^{z}_{L+1} \coloneqq -\sigma^{z}_{1}  \; .
\] 

\paragraph{Transfer matrix.} The hamiltonian (\ref{ham}) corresponds to the logarithmic derivative
of a six-vertex model transfer matrix with a particular boundary twist.
Let $\mathbb{V}_{\mathcal{A}} \cong \mathbb{V}_j \cong \mathbb{C}^2$ and consider the element
$G_{\mathcal{A}} \coloneqq \left( \begin{matrix} 0 & 1 \cr 1 & 0 \end{matrix} \right) \in \mbox{End}(\mathbb{V}_{\mathcal{A}})$.
Then consider the operator $\mathcal{R}_{\mathcal{A} j} \colon \mathbb{C} \rightarrow \mbox{End}(\mathbb{V}_{\mathcal{A}} \otimes \mathbb{V}_j)$ 
and define the transfer matrix $T \colon \mathbb{C} \rightarrow \mbox{End}(\mathbb{V}_{\mathcal{Q}})$ as
\[
\label{tmat}
T(\lambda) \coloneqq \mbox{Tr}_{\mathcal{A}} [ G_{\mathcal{A}} \mathop{\overrightarrow\prod}\limits_{1 \le j \le L } \mathcal{R}_{\mathcal{A} j} (\lambda - \mu_j) ]  \; , 
\] 
where $\lambda , \mu_j \in \mathbb{C}$. The trace in (\ref{tmat}) is taken over the space $\mathbb{V}_{\mathcal{A}}$ while the matrix 
$\mathcal{R} \in \mbox{End}(\mathbb{V}_1 \otimes \mathbb{V}_2)$ reads
\<
\label{rmat}
\mathcal{R}(\lambda) \coloneqq \left( \begin{matrix} a(\lambda) & 0 & 0 & 0 \cr 
0 & b(\lambda) & c(\lambda) & 0  \cr 0 & c(\lambda) & b(\lambda) & 0 \cr 0 & 0 & 0 & a(\lambda) \end{matrix}  \right) \; .
\>  
The non-null entries of (\ref{rmat}) corresponds to the functions: $a(\lambda) \coloneqq \sinh{(\lambda + \gamma)}$, 
$b(\lambda) \coloneqq \sinh{(\lambda)}$ and $c(\lambda) \coloneqq \sinh{(\gamma)}$ for parameters $\lambda , \gamma \in \mathbb{C}$.
In this way the hamiltonian (\ref{ham}) is obtained from the relation $\mathcal{H} \sim \left. \frac{d}{d \lambda} \ln{T(\lambda)} \right|_{\stackrel{\lambda = 0}{\mu_j = 0}}$
and it is also worth mentioning that the $\mathcal{R}$-matrix (\ref{rmat}) satisfies the standard Yang-Baxter equation \cite{Baxter_book}.
In addition to that the matrix $G_{\mathcal{A}}$ fulfills the property $[ \mathcal{R}, G_{\mathcal{A}} \otimes G_{\mathcal{A}}  ] = 0$
ensuring that the transfer matrix (\ref{tmat}) forms a commutative family.

\paragraph{Monodromy matrix.} Let $\mathcal{T}_{\mathcal{A}} \colon \mathbb{C} \rightarrow \mbox{End}( \mathbb{V}_{\mathcal{A}} \otimes \mathbb{V}_{\mathcal{Q}} )$ be 
the following operator
\[
\label{mono}
\mathcal{T}_{\mathcal{A}} (\lambda) \coloneqq \mathop{\overrightarrow\prod}\limits_{1 \le j \le L } \mathcal{R}_{\mathcal{A} j} (\lambda - \mu_j) \; ,
\] 
which shall be referred to as monodromy matrix. As the $\mathcal{R}$-matrix (\ref{rmat}) satisfies the Yang-Baxter equation, one can show
that the monodromy matrix (\ref{mono}) fulfills the following quadratic identity 
\[
\label{yba}
\mathcal{R}_{12} (\lambda_1 - \lambda_2) \mathcal{T}_1 (\lambda_1) \mathcal{T}_2 (\lambda_2) = \mathcal{T}_2 (\lambda_2) \mathcal{T}_1 (\lambda_1) \mathcal{R}_{12} (\lambda_1 - \lambda_2) \; .
\]
The relation (\ref{yba}) is usually referred to as Yang-Baxter algebra and since $\mathbb{V}_{\mathcal{A}} \cong \mathbb{C}^2$,
the monodromy matrix $\mathcal{T}_{\mathcal{A}}$ can be recasted as 
\[
\label{abcd}
\mathcal{T}_{\mathcal{A}} (\lambda) \eqqcolon \left( \begin{matrix} A(\lambda) & B(\lambda) \cr C(\lambda) & D(\lambda) \end{matrix} \right) \; ,
\]
with operators $A, B , C, D \in \mbox{End} ( \mathbb{V}_{\mathcal{Q}} )$. In this way we find that the transfer matrix (\ref{tmat}) simply
reads $T(\lambda) = B(\lambda) + C(\lambda)$.

\paragraph{Domain wall boundaries.} The $\mathcal{R}$-matrix (\ref{rmat}) encodes the statistical weights of
a six-vertex model as discussed in \cite{Baxter_book}. However, one still needs to define appropriate
boundary conditions in order to having a non-trivial partition function for the model. 
The case of DWBC for the six-vertex model was introduced on a square lattice of dimensions $L \times L$
by Korepin in \cite{Korepin_1982}. More precisely, in \cite{Korepin_1982} the author derives a recurrence relation
for the partition function of the model which was subsequently solved by Izergin \cite{Izergin_1987}.
The aforementioned partition function then reads
\[
\label{pf}
Z(\lambda_1, \dots , \lambda_L) = \bra{\bar{0}} \mathop{\overrightarrow\prod}\limits_{1 \le j \le L } B(\lambda_j) \ket{0} \; ,
\] 
with vectors $\ket{0}$ and $\ket{\bar{0}}$ defined as
\[
\label{states}
\ket{0} \coloneqq \left( \begin{matrix} 1 \cr 0 \end{matrix} \right)^{\otimes L} \qquad \mbox{and} \qquad \ket{\bar{0}} \coloneqq \left( \begin{matrix} 0 \cr 1 \end{matrix} \right)^{\otimes L} \; .
\]

\paragraph{Highest/lowest weight vectors.} The vectors $\ket{0}$ and $\ket{\bar{0}}$ defined in (\ref{states}) are respectively the
$\alg{sl}(2)$ highest and lowest weight vectors. The action of the entries of the monodromy matrix (\ref{abcd}) on those vectors
are given as follows:
\begin{align}
\label{action}
A(\lambda) \ket{0} &= \prod_{j=1}^{L} a(\lambda - \mu_j) \ket{0} &  D(\lambda) \ket{0} &= \prod_{j=1}^{L} b(\lambda - \mu_j) \ket{0} \nonumber \\
A(\lambda) \ket{\bar{0}} &= \prod_{j=1}^{L} b(\lambda - \mu_j) \ket{\bar{0}} & D(\lambda) \ket{\bar{0}} &= \prod_{j=1}^{L} a(\lambda - \mu_j) \ket{\bar{0}}  \nonumber \\
B(\lambda) \ket{\bar{0}} &= 0 & C(\lambda) \ket{0} &= 0 \; . \nonumber \\
\end{align}

\section{Functional equations}
\label{sec:FN}

The spectrum of the anti-periodic Heisenberg chain hamiltonian (\ref{ham}) can be obtained directly from the
spectrum of the transfer matrix (\ref{tmat}). This is due to the fact that the hamiltonian $\mathcal{H}$
is given by the logarithmic derivative of the transfer matrix $T$, in addition to the property $[ T(\lambda) , T(\mu)] = 0$
ensured by the relation (\ref{yba}). Thus one can shift the attention to the spectral problem associated with the transfer matrix $T$.
In its turn, this problem can be tackled through the method introduced in \cite{Galleas_2008} and subsequently extended 
in \cite{Galleas_2010, Galleas_2011, Galleas_2012, Galleas_SCP}. For that it is convenient to introduce some extra
definitions and conventions.

\begin{mydef} Let $B(\lambda_i) \in \mathbb{V}_{\mathcal{Q}}$ be an off-diagonal element of the Yang-Baxter
algebra as defined in (\ref{yba}). We then introduce the following notation for the product of $n$ 
generators $B$,
\[
\label{phi}
\left[ \lambda_1, \dots , \lambda_n \right]  \coloneqq \mathop{\overrightarrow\prod}\limits_{1 \le i \le n } B(\lambda_i) \; .
\]
\end{mydef}
 
\begin{remark} The property $B(\lambda) B(\mu) = B(\mu) B(\lambda)$, encoded in the relation (\ref{yba}),
ensures that $\left[ \lambda_1, \dots , \lambda_n \right]$ is symmetric under the permutation of variables
$\lambda_i \leftrightarrow \lambda_j$. Thus, when it is convenient, we shall also employ the simplified notation 
$\left[ X^{1,n} \right] \coloneqq \left[ \lambda_1, \dots , \lambda_n \right]$ where $X^{i,j} \coloneqq \{ \lambda_k \; | \; i \leq k \leq j \}$.
\end{remark}

\medskip

Next we recall that $T(\lambda) = B(\lambda) + C(\lambda)$ and consider the action of $T(\lambda_0)$ over the 
element $\left[ X^{1,n} \right]$. For that the most lengthy computation is the term $C(\lambda_0) \left[ X^{1,n} \right]$
which can be evaluated with the help of the commutation relations contained in (\ref{yba}). 
This computation has been performed in \cite{Korepin_1982, Galleas_2010} and we shall restrict ourselves to presenting
only the final results. In this way we are left with the following expression,
\<
\label{tphi}
T(\lambda_0) \left[ X^{1,n} \right] &=& \left[ X^{0,n} \right] + \sum_{1 \leq i \leq n} \left[ X^{1,n}_i \right] ( \Gamma_{0,i}^{i} A(\lambda_0) D(\lambda_i) + \Gamma_{i,0}^{i} A(\lambda_i) D(\lambda_0) ) \nonumber \\
&&+ \sum_{1 \leq i < j \leq n} \left[ X^{0,n}_{i,j} \right] ( \Omega_{i,j} A(\lambda_i) D(\lambda_j) + \Omega_{j,i} A(\lambda_j) D(\lambda_i) ) \; ,
\>
where $X^{1,n}_i \coloneqq X^{1,n} \backslash \{ \lambda_i \}$ and $X^{0,n}_{i,j} \coloneqq X^{0,n} \backslash \{ \lambda_i , \lambda_j \}$.
The coefficients $\Gamma_{j,k}^{i}$ and $\Omega_{i,j}$ in (\ref{tphi}) explicitly read
\<
\label{mn}
\Gamma_{j,k}^{i} &\coloneqq& \frac{c(\lambda_k - \lambda_j)}{b(\lambda_k - \lambda_j)} \prod_{ \lambda \in X^{1,n}_i   } \frac{a(\lambda_k - \lambda)}{b(\lambda_k - \lambda)} \frac{a(\lambda - \lambda_j)}{b(\lambda - \lambda_j)} \nonumber \\
\Omega_{i,j} &\coloneqq& \frac{c(\lambda_j - \lambda_0)}{a(\lambda_j - \lambda_0)} \frac{c(\lambda_0 - \lambda_i)}{a(\lambda_0 - \lambda_i)} \frac{a(\lambda_j - \lambda_i)}{b(\lambda_j - \lambda_i)}  \prod_{\lambda \in X^{0,n}_{i,j}} \frac{a(\lambda_j - \lambda)}{b(\lambda_j - \lambda)} \frac{a(\lambda - \lambda_i)}{b(\lambda - \lambda_i)} \; .
\>

The interpretation of the Yang-Baxter algebra as a source of functional equations \cite{Galleas_proc} can now be immediately
invoked. In this way one can recognize (\ref{tphi}) as a Yang-Baxter relation of order $n+1$ and, in addition to extra properties,
this relation will allow us to derive a functional equation describing the spectrum of the
transfer matrix $T$. For that it is also convenient to introduce the following definition.

\medskip
\begin{mydef}
Let $n \in \mathbb{Z}_{>0}$ be a discrete index and $\mathcal{M}(\lambda) \coloneqq \{ A, B, C, D \}(\lambda)$. 
Also, let $\mathcal{W}_n \coloneqq \mathcal{M}(\lambda_1) \times \mathcal{M}(\lambda_2) \times \dots \mathcal{M}(\lambda_n)$ 
with $n$-tuples $(\xi_1, \xi_2 , \dots , \xi_n)$ being understood as $\mathop{\overrightarrow\prod}\limits_{1 \le i \le n } \xi_i$.
Then consider the function space $\mathbb{C}[\lambda_1^{\pm 1} , \dots , \lambda_n^{\pm 1} ]$ of regular complex-valued functions
on $(\lambda_1 , \dots , \lambda_n ) \in \mathbb{C}^n$ and define $\tilde{\mathcal{W}}_n \coloneqq \mathbb{C}[\lambda_1^{\pm 1} , \dots , \lambda_n^{\pm 1} ] \otimes \gen{span}_{\mathbb{C}}(\mathcal{W}_n)$.
The map $\pi_n$ is then introduced as the following $n$-additive continuous map
\[
\label{mapi}
\pi_n \colon \tilde{\mathcal{W}}_n \rightarrow \mathbb{C}[\lambda_1^{\pm 1} , \lambda_2^{\pm 1} , \dots , \lambda_n^{\pm 1} ] \; .
\]
In other words, the map $\pi_n$ associates a multivariate complex function to any product of $n$ generators of the Yang-Baxter algebra.
\end{mydef}

\medskip
The next step within this approach consists in finding a suitable realization of the map $\pi_n$ which is able to convert (\ref{tphi})
into appropriate functional equations. We shall proceed along the lines of \cite{Galleas_proc} and adopt a particular scalar product
as realization of $\pi_n$.

\paragraph{Realization of $\pi_n$.} Let $\ket{\gen{\Psi}} \in \mbox{span}(\mathbb{V}_{\mathcal{Q}})$ be an eigenvector
of the transfer matrix (\ref{tmat}) with eigenvalue $\Lambda(\lambda)$. More precisely, we have the action 
$T(\lambda) \ket{\gen{\Psi}} = \Lambda(\lambda) \ket{\gen{\Psi}}$. Then, taking into account the definition (\ref{states}),
we define the map $\pi$ as 
\[
\label{pir}
\pi_{n+1} (\mathcal{A}) \coloneqq \bra{\gen{\Psi}} \mathcal{A} \ket{0} \qquad \qquad \forall \; \mathcal{A} \; \in \; \tilde{\mathcal{W}}_{n+1} \; . 
\]

\bigskip 
At this stage we have gathered all the ingredients required to convert the Yang-Baxter algebra relation (\ref{tphi})
into a functional equation characterizing the eigenvalues $\Lambda$. 
For that we only need to apply the map (\ref{mapi}) to the relation (\ref{tphi}), taking into account the realization (\ref{pir}). 
This procedure can be effectively carried out by  noticing that the LHS of (\ref{tphi}) obeys the following reduction
property $\pi_{n+1} \rightarrow \pi_{n}$,
\[
\label{piLHS}
\pi_{n+1} ( T(\lambda_0) \left[ X^{1,n} \right] ) = \Lambda(\lambda_0) \pi_{n} ( \left[ X^{1,n} \right] ) \; .
\]
On the other hand, by applying the map (\ref{mapi}) to (\ref{tphi}), the terms in the RHS are of the following form:
$\pi_{n+1}( \left[ X^{0,n} \right] )$, $\pi_{n+1}( \left[ X_i^{1,n} \right] A(z_1) D(z_2) )$ for $(z_1, z_2) \in \gen{Sym}( \{ \lambda_0 , \lambda_i \})$
and $\pi_{n+1}( \left[ X_{i,j}^{0,n} \right] A(z_1) D(z_2) )$ for $(z_1, z_2) \in \gen{Sym}( \{ \lambda_i , \lambda_j \})$.
The term $\pi_{n+1}( \left[ X^{0,n} \right] )$ can not be significantly simplified, but for 
$\gen{Y} \in \{ X_i^{1,n} , X_{i,j}^{0,n} \}$ and using (\ref{action}), we find that
\[
\label{piRHS}
\pi_{n+1}( \left[ \gen{Y} \right] A(z_1) D(z_2) ) = \prod_{k=1}^L a(z_1 - \mu_k) b(z_2 - \mu_k) \; \pi_{n-1}( \left[ \gen{Y} \right] ) \; .
\]

Our results so far can be written in a more convenient form with the help of the notation 
$\pi_n (\left[X \right]) \coloneqq \mathcal{F}_n (X )$ for a given set $X = \{ \lambda_k \}$ of cardinality $n$.
In this way, this procedure yields the following set of functional equations,
\<
\label{FL}
\Lambda(\lambda_0) \mathcal{F}_n ( X^{1,n} ) &=& \mathcal{F}_{n+1} ( X^{0,n} ) + \sum_{1 \leq i \leq n} M_i^{(n)} \mathcal{F}_{n-1} ( X^{1,n}_i ) \nonumber \\
&& + \sum_{1 \leq i < j \leq n} N_{j,i}^{(n)} \mathcal{F}_{n-1} ( X^{0,n}_{i,j} ) \; ,
\>
with coefficients $M_i^{(n)} = M_i^{(n)}( \vec{X}^{0,n} )$ and $N_{j,i}^{(n)} = N_{j,i}^{(n)}( \vec{X}^{0,n} )$ given by
\<
\label{coeff}
M_i^{(n)} &\coloneqq& \Gamma_{0,i}^i \prod_{k=1}^{L} a(\lambda_0 - \mu_k) b(\lambda_i - \mu_k) + \Gamma_{i,0}^i \prod_{k=1}^{L} a(\lambda_i - \mu_k) b(\lambda_0 - \mu_k) \nonumber \\
N_{j,i}^{(n)} &\coloneqq& \Omega_{i,j} \prod_{k=1}^{L} a(\lambda_i - \mu_k) b(\lambda_j - \mu_k) + \Omega_{j,i} \prod_{k=1}^{L} a(\lambda_j - \mu_k) b(\lambda_i - \mu_k) \; .
\>
The symbol $\vec{X}^{0,n}$ has been introduced in order to emphasize that the functions $M_i^{(n)}$ and $N_{j,i}^{(n)}$
are not invariant under the permutation of all variables. Its precise definition is given as follows.
\begin{mydef}
Let $i,j \in \mathbb{Z}$ such that $i<j$. Then $\vec{X}^{i,j}$ stands for the vector
\[
\vec{X}^{i,j} \coloneqq (\lambda_i , \lambda_{i+1} , \lambda_{i+2} , \dots , \lambda_{j} ) \; .
\]
For latter convenience we shall also define the symbols $\vec{X}_k^{i,j}$ and $\vec{X}_{k,l}^{i,j}$
for $i \leq k,l \leq j$ such that $k<l$. They are respectively defined as
\<
\vec{X}_k^{i,j} &\coloneqq& (\lambda_i , \lambda_{i+1} , \dots,  \lambda_{k-1} , \lambda_{k+1} , \dots , \lambda_{j} ) \nonumber \\
\vec{X}_{k,l}^{i,j} &\coloneqq& (\lambda_i , \lambda_{i+1} , \dots,  \lambda_{k-1} , \lambda_{k+1} , \dots , \lambda_{l-1} , \lambda_{l+1} , \dots , \lambda_{j} ) \; .
\>
\end{mydef}

Some remarks are required at this stage. For instance, the functional equation (\ref{FL}) consists
of an extension of the equation obtained in \cite{Galleas_2008} and it has also been recently described
in \cite{WengLi_2013}. Moreover, at algebraic level there is no upper limit for the discrete index $n$ in 
Eq. (\ref{FL}). However, the $\alg{sl}(2)$ highest weight representation theory imposes
an upper bound for that index.

\section{The eigenvalues $\Lambda$ and the partition function $Z$}
\label{sec:SOL}

In the previous section we have derived a functional equation involving the eigenvalues
$\Lambda$ of the transfer matrix (\ref{tmat}) and a certain set of functions 
$\mathcal{F}_n$. Here we intend to show that the functions $\mathcal{F}_n$ can be eliminated
from the system of Eqs. (\ref{FL}) yielding a single equation for the eigenvalues $\Lambda$.
The functional equation for $\Lambda$ obtained in this way will depend explicitly on the partition
function $Z$ of the six-vertex model with DWBC. 

\paragraph{Highest weight and domain walls.} The highest weight representation theory of
the $\alg{sl}(2)$ algebra gives an upper bound for the number of operators $B$ entering the product
(\ref{phi}) as discussed in \cite{Korepin_1982, Galleas_2010}. 
This feature is manifested in the following property,
\[
\label{high}
[ X^{1,L} ]  \ket{0} = Z(X^{1,L}) \ket{\bar{0}} \; .
\]
Thus the relations (\ref{high}) and (\ref{pir}) imply that $\mathcal{F}_{L} ( X^{1,L} ) = Z(X^{1,L}) \bar{\mathcal{F}}_0$ 
where $\bar{\mathcal{F}}_0 = \left< \gen{\Psi} \right| \left. \bar{0} \right>$.
Moreover, the function $\mathcal{F}_n$ vanishes for $n > L$ due to (\ref{action}), (\ref{pir}) and (\ref{high}).
It is worth remarking here that Eq. (\ref{FL}) assumes that $\mathcal{F}_n$ vanishes for $n < 0 $.

\bigskip
In order to illustrate our procedure, let us firstly have a closer look at Eq. (\ref{FL}) for the case $L=2$.
In that case we can set $n=0,1,2,3$ and by doing so we are left with the following set of equations:
\<
\label{L2}
\Lambda(\lambda_0) \mathcal{F}_0 &=& \mathcal{F}_1 (X^{0,0}) \nonumber \\
\Lambda(\lambda_0) \mathcal{F}_1 (X^{1,1}) &=& Z( X^{0,1} ) \bar{\mathcal{F}}_0 + M_1^{(1)} \mathcal{F}_0 \nonumber \\
\Lambda(\lambda_0) Z( X^{1,2} ) \bar{\mathcal{F}}_0 &=& \sum_{1 \leq i \leq 2} M_i^{(2)} \mathcal{F}_1 (X^{1,2}_i) + N_{2,1}^{(2)} \mathcal{F}_1 (X^{0,2}_{1,2}) \nonumber \\
0 &=& \sum_{1 \leq i \leq 3} M_i^{(3)} Z( X^{1,3}_i ) + \sum_{1 \leq i < j \leq 3} N_{j,i}^{(3)} Z( X^{0,3}_{i,j} ) \; .
\>
The last equation in (\ref{L2}) involves solely the partition function $Z$ and it had been previously described in 
\cite{Galleas_2010}. This single equation is fully able to determine the function $Z$, up to an overall constant factor,
while the remaining equations relate the eigenvalue $\Lambda$ and the auxiliary function $\mathcal{F}_1$.
We shall then use the first equation of (\ref{L2}) to eliminate $\mathcal{F}_1$ from the second and third
equations. By doing so we are left with the relations,
\<
\label{LL2}
\Lambda(\lambda_0) \Lambda(\lambda_1) &=& Z(X^{0,1}) k_0 + M_1^{(1)} \nonumber \\
\Lambda(\lambda_0) \left[ Z(X^{1,2}) k_0 - N_{2,1}^{(2)} \right]  &=& M_1^{(2)} \Lambda(\lambda_2) + M_2^{(2)} \Lambda(\lambda_1) \; ,
\>
where $k_0 = \bar{\mathcal{F}}_0 /\mathcal{F}_0$. In what follows we shall assume that the partition function $Z$ is already determined 
and the only unknown factors in (\ref{LL2}) are the coefficients $k_0$ and the function $\Lambda$.
Both equations in (\ref{LL2}) are able to determine $k_0$ and the eigenvalues $\Lambda$, however, the first equation is non-linear
while the second is linear. In fact, the direct inspection of (\ref{LL2}) reveals that 
$k_0 = \Lambda(\mu_1) \Lambda(\mu_2) \left[ c^2 a(\mu_1 - \mu_2) a(\mu_2 - \mu_1) \right]^{-1}$.

For the case $L=3$ we can set $n=0,1,2,3,4$ in (\ref{FL}). Each choice produces an independent equation and the whole set 
consists of the following equations,
\<
\label{L3}
\Lambda(\lambda_0) \mathcal{F}_0 &=& \mathcal{F}_1 (X^{0,0}) \nonumber \\
\Lambda(\lambda_0) \mathcal{F}_1 (X^{1,1}) &=& \mathcal{F}_2 (X^{0,1}) + M_1^{(1)} \mathcal{F}_0 \nonumber \\
\Lambda(\lambda_0) \mathcal{F}_2 (X^{1,2}) &=& Z(X^{0,2}) \bar{\mathcal{F}}_0 + \sum_{1 \leq i \leq 2} M_i^{(2)} \mathcal{F}_1 (X^{1,2}_i) + N_{2,1}^{(2)} \mathcal{F}_1 (X^{0,2}_{1,2}) \nonumber \\
\Lambda(\lambda_0) Z(X^{1,3}) \bar{\mathcal{F}}_0 &=&  \sum_{1 \leq i \leq 3} M_i^{(3)} \mathcal{F}_2 (X^{1,3}_i) + \sum_{1 \leq i < j  \leq 3} N_{j,i}^{(3)} \mathcal{F}_2 (X^{0,3}_{i,j}) \nonumber \\
0 &=& \sum_{1 \leq i \leq 4} M_i^{(4)} Z( X^{1,4}_i ) + \sum_{1 \leq i < j \leq 4} N_{j,i}^{(4)} Z( X^{0,4}_{i,j} ) \; .
\>
Similarly to the case $L=2$, we can eliminate the functions $\mathcal{F}_i$ from the system of equations (\ref{L3})
in favor of the function $\Lambda$. By carrying out this procedure recursively, we find that the third and fourth equations in (\ref{L3})
can be rewritten as 
\<
\label{LL3}
\Lambda(\lambda_0) \Lambda(\lambda_1) \Lambda(\lambda_2) &=& Z(X^{0,2}) k_0 + \Lambda(\lambda_0) [ M_1^{(1)} (\vec{X}^{1,2}) + N_{2,1}^{(2)} (\vec{X}^{0,2}) ] \nonumber \\
&&+ \; \Lambda(\lambda_1) M_2^{(2)} ( \vec{X}^{0,2} ) + \Lambda(\lambda_2) M_1^{(2)} ( \vec{X}^{0,2} ) \nonumber \\
\Lambda(\lambda_0) Z(X^{1,3}) k_0 &=& \sum_{1 \leq i \leq 3} M_i^{(3)}(\vec{X}^{0,3}) \prod_{\stackrel{k=1}{k \neq i}}^3 \Lambda(\lambda_k) + \sum_{1 \leq i < j \leq 3} N_{j,i}^{(3)}(\vec{X}^{0,3}) \Lambda(\lambda_0) \prod_{\stackrel{k=1}{k \neq i,j}}^3 \Lambda(\lambda_k) \nonumber \\
&& - \; \sum_{1 \leq i \leq 3} M_i^{(3)} (\vec{X}^{0,3}) M_1^{(1)} (\vec{X}_i^{1,3}) - \sum_{1 \leq i < j \leq 3} N_{j,i}^{(3)} (\vec{X}^{0,3}) M_1^{(1)} (\vec{X}_{i,j}^{0,3}) \; . \nonumber \\
\>
In contrast to the case $L=2$, none of the equations in (\ref{LL3}) is linear and this is the general behavior for arbitrary $L$.
At this stage it is also worth stressing that for both cases, namely $L=2$ and $L=3$, we have explicitly written two sets of equations
relating the eigenvalues $\Lambda$ and the partition function $Z$. Each set is formed by two equations, i.e. (\ref{LL2}) and (\ref{LL3}), however,
there is a dramatic difference between the first and the second equations of each set.
For instance, while the first equation runs over the set of variables $\{ \lambda_k \; | \; 0 \leq k \leq L-1 \}$,
the second equation is defined over the set $\{ \lambda_k \; | \; 0 \leq k \leq L \}$.
The partition function $Z$ can also be described through functional equations and a similar feature had previously
appeared for that problem. If we compare the functional equations for $Z$ derived in \cite{Galleas_2010} and \cite{Galleas_2012}
we can readily see they are defined over a different number of variables.
In what follows we shall focus on the functional equation relating $\Lambda$ and $Z$ generalizing the first equation
of (\ref{LL2}) and (\ref{LL3}) for arbitrary values of $L$. The following definitions will be helpful.

\begin{mydef}
The symbol $[x]$ is defined as 
\<
\left[x \right] \coloneqq \begin{cases}
x \quad \quad \quad \mbox{for} \;\; x \in 2 \mathbb{Z}_{>0} \\
x-1 \quad \; \mbox{for} \;\; x \in (2 \mathbb{Z}_{>0} +1) \quad . \end{cases}
\>
\end{mydef}
\medskip
\begin{mydef}
Let $f(\lambda) \in \mathbb{C}[\lambda]$ and consider the product operator $\widehat{\prod}_{\lambda}^{i_1, \dots , i_{2m}} \colon \mathbb{C}[\lambda] \rightarrow \mathbb{C}[ \lambda_{v_1} ] \times \dots \times \mathbb{C} [\lambda_{v_{L-2m}}]$
for $\lambda_{v_j} \in X^{0,L-1}_{i_1 , \dots , i_{2m}}$ such that $\lambda_{v_j} \neq \lambda_{v_k}$ if $j \neq k$.
The relation $X^{l,m}_{i_1 , \dots, i_{n+1}} = X^{l,m}_{i_1 , \dots, i_{n}} \backslash \{ \lambda_{i_{n+1}} \}$ generalizes
recursively our previous definition and the product operator $\widehat{\prod}_{\lambda}^{i_1, \dots , i_{2m}}$ is defined as
\<
\left( \widehat{\prod}_{\lambda}^{i_1, \dots , i_{2m}} f  \right) (\lambda) &\coloneqq& \prod_{\lambda \in X^{0,L-1}_{i_1 , \dots, i_{2m}}} f(\lambda) = \prod^{L-1}_{\stackrel{k=0}{k \neq i_1, \dots , i_{2m}}} f(\lambda_k) \; .
\>
\end{mydef}

\begin{theorem} \label{ZV}
The partition function $Z$ can be written in terms of the eigenvalues $\Lambda$ according to the formula
\<
\label{Lgen}
Z(X^{0,L-1}) k_0 = \left\{ \sum_{m=0}^{[L]/2} \; \sum_{0 \leq i_1 < \dots < i_{2m} \leq L-1} V^{(2m)}_{i_{2m} , \dots, i_1} \; \widehat{\prod}_{\lambda}^{i_1, \dots , i_{2m}} \right\}  \Lambda(\lambda)  \; ,
\>
where 
\<
\label{VV}
V^{(2m)}_{i_{2m} , \dots , i_1} &\coloneqq&  \sum_{J} \prod_{l=1}^{m} \prod_{n=1}^{L} a(\lambda_{j_l} - \mu_n ) \prod_{\lambda \in X^{0,L-1}_{i_1 , \dots , i_{2m}}} \frac{a(\lambda - \lambda_{j_l})}{b(\lambda - \lambda_{j_l})} \nonumber \\
&& \qquad \qquad \times \sum_{K(J)} \prod_{l=1}^{m} \prod_{n=1}^{L} b(\lambda_{k_l} - \mu_n ) \frac{c(\lambda_{j_l} - \lambda_{k_l})}{b(\lambda_{j_l} - \lambda_{k_l})} \prod_{\lambda \in X^{0,L-1}_{i_1 , \dots , i_{2m}}} \frac{a(\lambda_{k_l} - \lambda)}{b(\lambda_{k_l} - \lambda)}  \nonumber \\
&& \qquad \qquad \qquad \times \; \; \prod_{1 \leq r < s \leq m} \frac{a(\lambda_{k_r} - \lambda_{k_s})}{b(\lambda_{k_r} - \lambda_{k_s})} \frac{a(\lambda_{k_r} - \lambda_{j_s})}{b(\lambda_{k_r} - \lambda_{j_s})} \frac{a(\lambda_{k_s} - \lambda_{j_r} + \gamma)}{b(\lambda_{k_s} - \lambda_{j_r})} \; ,  \nonumber \\
\>
with summation symbols defined as $\displaystyle \sum_{J} \coloneqq \sum_{\stackrel{j_1 , \dots , j_m \in \mathcal{I}_{2m}}{j_1 < j_2 < \dots < j_m}}$
and $\displaystyle \sum_{K(J)} \coloneqq \sum_{\stackrel{k_1 , \dots , k_m \in \mathcal{J}_{2m}}{k_{\alpha} \neq k_{\beta}}}$. 
The symbols $\mathcal{I}_{2m}$ and $\mathcal{J}_{2m}$ stand respectively for the sets $\mathcal{I}_{2m} \coloneqq \{ i_1, \dots , i_{2m} \}$ and 
$\mathcal{J}_{2m} \coloneqq \mathcal{I}_{2m} \backslash \{ j_1, \dots , j_{m} \}$. For clarity's sake, we stress here that $V^{(0)} \coloneqq 1$.
\end{theorem}
\begin{proof}
The proof follows from the extension of the derivation presented for the cases $L=2$ and $L=3$. These cases
are respectively covered by formulae (\ref{LL2}) and (\ref{LL3}). The derivation of formula (\ref{VV}) for arbitrary
$L$ is discussed in \appref{sec:proof}. 
\end{proof}

\medskip
\begin{example}
The RHS of (\ref{Lgen}) for $L=2$ reads
\[
\label{ZB2}
\Lambda(\lambda_0) \Lambda(\lambda_1) + V^{(2)}_{1, 0} \; ,
\]
while for $L=3$ we have
\[
\label{ZB3}
\Lambda(\lambda_0) \Lambda(\lambda_1) \Lambda(\lambda_2) + \sum_{0 \leq i_1 < i_2 \leq 2} V_{i_2 , i_1}^{(2)}  \prod_{\stackrel{k=0}{k \neq i_1 , i_2}}^{2} \Lambda(\lambda_k) \; .
\]
\end{example}

The structure of the relation (\ref{Lgen}) is quite appealing and some remarks are in order. For instance, the relation (\ref{Lgen})
converts the problem of evaluating the partition function $Z$ into the diagonalisation of the transfer matrix of the six-vertex
model with a non-diagonal boundary twist. This situation is analogous to the case of the six-vertex model defined on a torus 
where the model partition function is given in terms of the eigenvalues of the standard transfer matrix of the six-vertex model
with periodic boundary conditions \cite{Baxter_book, deGier_Galleas11}. Furthermore, if one assumes that the eigenvalues $\Lambda$
are parameterized by solutions of Bethe ansatz like equations, as obtained in \cite{Baxter_1995, Niccoli_2013}, then we could
expect the relation (\ref{Lgen}) to offer access to thermodynamic properties of the six-vertex model with DWBC in the same fashion
as for the case with toroidal boundary conditions.

\subsection{The zeroes $w_j$}
\label{sec:WJ}

The eigenvalues $\Lambda(\lambda)$ are essentially a polynomial of order $L-1$ in the variable $x \coloneqq e^{2 \lambda}$ 
as demonstrated in \cite{Galleas_2008}. Thus it can be written in terms of its zeroes $w_j$ as
\[
\label{wj}
\Lambda (\lambda) = \Lambda(0) \prod_{j=1}^{L-1} \frac{\sinh{(w_j - \lambda)}}{\sinh{(w_j)}} \; .
\]
We shall assume that the zeroes $w_j$ are all distinct and the following corollary will allow
us to determine the zeroes $w_j$.

\begin{corollary}
The relation (\ref{Lgen}) under the specialization $\lambda_j = w_j$ for $1 \leq j \leq L-1$
implies the following constraints,
\<
\label{LZ01}
\frac{Z(\lambda_0, w_1, \dots , w_{L-1})}{V^{([L])}_{L-1, \dots, 0}(\lambda_0, w_1, \dots , w_{L-1})} k_0 = \begin{cases}
(-1)^{\frac{L}{2}}  \qquad \mbox{for} \quad L \in 2 \mathbb{Z}_{>0} \cr
\Lambda(\lambda_0)  \qquad \; \mbox{for} \quad L \in (2 \mathbb{Z}_{>0} +1) \; .
\end{cases}
\>
\end{corollary}

Now we can use the analytic properties of the functions $Z$ and $V^{([L])}_{L-1, \dots, 0}$, in addition to 
the relation (\ref{LZ01}), to determine the set of zeroes $\{ w_j \}$. For that it will
be important to notice that the variable $\lambda_0$ in (\ref{LZ01}) is still an arbitrary complex variable. 
Moreover, the partition function $Z$ is a symmetric multivariate polynomial \cite{Korepin_1982, Galleas_2011},
while the function $V^{([L])}_{L-1, \dots, 0}$ for $L \in 2 \mathbb{Z}_{>0}$ consists of a polynomial of order
$L-1$ in the variable $x_0 \coloneqq e^{2 \lambda_0}$. On the other hand, for $L \in 2 \mathbb{Z}_{>0} + 1$ we have
\[
\label{barV}
V^{(L-1)}_{L-1, \dots, 0} (\lambda_0, w_1, \dots , w_{L-1}) = \frac{\tilde{V}^{(L-1)}_{L-1, \dots, 0} (\lambda_0, w_1, \dots , w_{L-1})}{\prod_{j=1}^{L-1} b(\lambda_0 - w_j)} \; ,
\] 
where the function $\tilde{V}^{(L-1)}_{L-1, \dots, 0} (\lambda_0, w_1, \dots , w_{L-1})$ is a polynomial of order $L-1$ in the variable $x_0$.

\paragraph{The case $L \in 2 \mathbb{Z}_{>0}$.}  This case corresponds to even values of $L$ and we can use the analytical properties 
of (\ref{LZ01}) to characterize the set of variables $\{ w_j \}$. For this analysis it is important to recall
that both functions $Z$ and $V^{(L)}_{L-1, \dots, 0}$ in the LHS of (\ref{LZ01}) are polynomials of the same degree in the variable $x_0$. 
Then, since the RHS of (\ref{LZ01}) is a constant, we can conclude that the residues of the LHS must vanish at the zeroes
of $V^{(L)}_{L-1, \dots, 0}$. In other words, the zeroes of $Z$ and $V^{(L)}_{L-1, \dots, 0}$ must coincide.
This analysis yields a formal condition determining the set $\{ w_j \}$ which is summarized in Corollary \ref{LZVeven}.

\begin{corollary} \label{LZVeven}
Consider $L \in 2 \mathbb{Z}_{>0}$ and let $\lambda_k^{Z} \in \{ \lambda \in \mathbb{C} \; | \; Z(\lambda, w_1, \dots , w_{L-1}) = 0 \}$
while $\lambda_k^{V} \in \{ \lambda \in \mathbb{C} \; | \; V^{(L)}_{L-1, \dots, 0}(\lambda, w_1, \dots , w_{L-1}) = 0 \}$. 
The zeroes $\lambda_k^{Z}$ and $\lambda_k^{V}$ shall depend on the set of parameters $\{ w_j \}$ and we can conclude that
\[
\label{BAeven}
\lambda_k^{Z} (\{ w_j \})   = \lambda_k^{V} (\{ w_j \})  \qquad \qquad 1 \leq k \leq L-1 \; .
\]
\end{corollary}
\noindent The direct inspection of (\ref{BAeven}) for small values of $L$ reveals that the variables 
$w_j$ are completely fixed by the aforementioned constraints.

\bigskip

\paragraph{The case $L \in (2 \mathbb{Z}_{>0} + 1$).} The situation for $L$ odd requires a slightly more 
elaborated analysis due to the presence of the eigenvalue $\Lambda$ in the RHS of (\ref{LZ01}). In that case we
also need to consider (\ref{barV}), and it turns out that (\ref{LZ01}) simplifies to
\[
\label{ZV1}
\frac{\Lambda(0)}{k_0} \prod_{j=1}^{L-1} b(-w_j)^{-1} = \frac{Z(\lambda_0, w_1, \dots , w_{L-1})}{\tilde{V}^{(L-1)}_{L-1, \dots, 0}(\lambda_0, w_1, \dots , w_{L-1})} \; .
\]
The LHS of (\ref{ZV1}) is a constant with respect to the variable $\lambda_0$ while the RHS consists of a rational function.
Thus the polynomials $Z$ and $\tilde{V}^{(L-1)}_{L-1, \dots, 0}$ must share the same zeroes. Similarly to the $L$ even case,
this statement can be formulated more precisely as the following corollary.

\begin{corollary} \label{LZVodd}
Assume that $L \in (2 \mathbb{Z}_{>0} + 1)$ and let $\lambda_k^{Z} \in \{ \lambda \in \mathbb{C} \; | \; Z(\lambda, w_1, \dots , w_{L-1}) = 0 \}$ 
as previously defined. Also, let $\tilde{\lambda}_k^{V} \in \{ \lambda \in \mathbb{C} \; | \; \tilde{V}^{(L-1)}_{L-1, \dots, 0}(\lambda, w_1, \dots , w_{L-1}) = 0 \}$. 
Thus we have the following conditions determining the set of variables $\{ w_j \}$,
\[
\label{BAodd}
\lambda_k^{Z} (\{ w_j \})   = \tilde{\lambda}_k^{V} (\{ w_j \})  \qquad \qquad 1 \leq k \leq L-1 \; .
\]
\end{corollary} 

\medskip

Both Corollaries \ref{LZVeven} and \ref{LZVodd} state that the zeroes of the partition function $Z$, with
respect to one of its variables, coincide with the zeroes of the function $V^{([L])}_{L-1,\dots, 0}$ when
the remaining variables correspond to zeroes of the transfer matrix eigenvalues $\Lambda$.

\medskip 

\paragraph{Wronskian condition.} The constraints (\ref{BAeven}) and (\ref{BAodd}) are given in terms
of the zeroes of certain polynomials whose explicit evaluation might still be a very non-trivial
problem. Alternatively, one can also obtain equations determining the set of zeroes $\{ w_j \}$ in terms of the
coefficients of the polynomial part of $V^{([L])}_{L-1, \dots, 0}$ and $Z$.
This analysis can be performed for $L \in 2 \mathbb{Z}_{>0}$ and $L \in (2 \mathbb{Z}_{>0}+1)$ in an unified
manner with the help of the function $F$ defined as
\[
\label{FF}
F \coloneqq \begin{cases}
V^{(L)}_{L-1, \dots, 0} \qquad \quad \mbox{for} \; L \in 2 \mathbb{Z}_{>0} \cr 
\tilde{V}^{(L-1)}_{L-1, \dots, 0} \qquad \quad \mbox{for} \; L \in (2 \mathbb{Z}_{>0} + 1 )
\end{cases} \; .
\] 
The term $\tilde{V}^{(L-1)}_{L-1, \dots, 0}$ in (\ref{FF}) has been previously defined in (\ref{barV}).
Thus the function $F(\lambda_0, w_1 , \dots , w_{L-1})$ is a polynomial of order $L-1$ in the variable $x_0$ for $L \in \mathbb{Z}_{>0}$. 

The equations fixing the zeroes $w_j$ from the coefficients of $Z$ and $F$ can be directly read from the
relations (\ref{LZ01}) and (\ref{ZV1}). However, this approach would leave us with an overall constant factor
and we can avoid this drawback by simply demanding that the Wronskian determinant between $Z$ and $F$ vanishes. 
This is justified by the fact that (\ref{LZ01}) and (\ref{ZV1}) tell us that $Z$ and $F$ and two linearly dependent functions.
In this way we are left with the condition,
\[
\label{PW}
P(x_0) \coloneqq Z(x_0, \{ w_j \}) F^{\prime}(x_0, \{ w_j \}) - F(x_0, \{ w_j \}) Z^{\prime}(x_0, \{ w_j \}) = 0 \; ,
\]
where the symbol ($^{\prime}$) denotes differentiation with respect to the variable $x_0$.

The function $P$ is a polynomial of order $[L]$ in the variable $x_0$ which must vanish
in the entire complex domain according to the Wronskian condition (\ref{PW}). The coefficients of $P$
are given by 
\[
\mathcal{C}_k = \frac{1}{k !} \left. \frac{\partial^k P}{\partial x_0^k} \right|_{x_0 = 0} \; ,
\]
and we demand these coefficients to vanish in order to satisfy (\ref{PW}). Thus we end up with the following
formal condition fixing the zeroes $w_j$,
\[
\label{CK}
\mathcal{C}_k (\{ w_j \}) = 0 \qquad \qquad 0 \leq k \leq [L] \; .
\] 

It is important to remark here that (\ref{CK}) provides one or two more equations than
variables $w_j$ to be determined. For $L$ even we have one more equation, while for $L$ odd we have
two additional equations. Each equation is a non-linear algebraic equation and consequently we
have a large number of solutions. This feature is similar to what one finds when solving
standard Bethe ansatz equations. However, the direct inspection of the solutions of (\ref{CK}) for small
values of $L$ reveals that these extra equations play the role of a filter keeping only solutions which actually
describe the spectrum of the transfer matrix (\ref{tmat}).

\subsection{Truncation at roots of unity}
\label{sec:roots}

Vertex models based on solutions of the Yang-Baxter equation can exhibit special properties
when its anisotropy parameter satisfies certain root of unity conditions. For instance, 
Tarasov demonstrated in \cite{Tarasov_2005} that the property
\[
\label{rou}
\mathop{\overrightarrow\prod}\limits_{0 \le k \le l-1 } B(\lambda - k \gamma) = 0 
\]
holds for the $\mathcal{U}_q[\widehat{\alg{sl}}(2)]$ invariant six-vertex model when the anisotropy
parameter $\gamma$ obeys the condition $e^{2 l \gamma} =1$.
The case $l=1$ is not illuminating for our present discussion as we can see from definitions (\ref{mono})
and (\ref{abcd}) that both operators $B(\lambda)$ and $C(\lambda)$ are proportional
to the factor $(e^{2\gamma} - 1)$. Consequently the transfer matrix (\ref{tmat}) is  also proportional
to that same quantity, and this implies that its eigenvalues trivially vanishes when we set $l=1$.
The free-fermion point $\gamma = \ii \pi / 2$ is, in its turn, covered by the case $l=2$ and we shall start our analysis with this case
despite its triviality. We shall then consider the cases $l=3$ and $l=4$ separately before discussing the general case. 

\paragraph{$l=2$.} Our goal here is to analyze the system of equations (\ref{FL}) taking into account the representation
theoretic properties of the functions $\mathcal{F}_n (X^{1,n}) = \pi_n ( [ \lambda_1 , \dots , \lambda_n ] )$. The property (\ref{rou})
can then be used in a very natural way and for $l=2$ we have $[\lambda , \lambda - \gamma ] = 0$. 
More precisely, we can exploit this property by looking at (\ref{LL2}) under the specialization $\lambda_0 = \lambda$ and
$\lambda_1 = \lambda - \gamma$. By doing so we are left with the following functional relation,
\<
\label{r2}
\Lambda (\lambda ) \Lambda (\lambda - \gamma) &=& M_1^{(1)} (\lambda , \lambda - \gamma) \nonumber \\
&=&  \prod_{k=1}^L \sinh{(\lambda - \mu_k)}^2  - \prod_{k=1}^L \sinh{(\lambda - \mu_k + \gamma)} \sinh{(\lambda - \mu_k - \gamma)} \; . \nonumber \\
\>
It is worth remarking here that (\ref{r2}) corresponds to an analogous of the inversion relation proposed by Stroganov 
in \cite{Stroganov_1979}. Next we consider the representation (\ref{wj}) and set $\lambda = w_i$ in (\ref{r2}).
This procedure then yields the following equation determining the set of zeroes $\{ w_j \}$,
\[
\label{BAl2}
\prod_{k=1}^L \frac{\sinh{(w_i - \mu_k + \gamma)}}{\sinh{(w_i - \mu_k)}} \frac{\sinh{(w_i - \mu_k - \gamma)}}{\sinh{(w_i - \mu_k)}} = 1 \; .
\]
As previously mentioned, the free-fermion point $\gamma = \ii \pi /2$ fits in the case $l=2$ and at this particular point
we find that (\ref{BAl2}) simplifies to 
\[
\label{free}
\prod_{k=1}^L \coth{(w_i - \mu_k)}^2 = 1 \; .
\]
We can now readily see that (\ref{free}) generalizes the proposal of \cite{WengLi_2013} in the presence of inhomogeneities $\mu_k$.

\bigskip 

\paragraph{$l=3$.} In that case the property (\ref{rou}) reads $ [\lambda, \lambda - \gamma, \lambda - 2\gamma ] = 0$
and we can immediately substitute it in equations (\ref{L3}) and (\ref{LL3}) under the specializations $\lambda_j = \lambda - j\gamma$.
By doing so we are left with the following relation,
\<
\label{r3}
\Lambda(\lambda) \Lambda(\lambda - \gamma) \Lambda(\lambda - 2\gamma) &=& \Lambda(\lambda) \left[ M_1^{(1)} (\lambda - \gamma, \lambda - 2\gamma) 
+ N_{2,1}^{(2)} (\lambda, \lambda - \gamma, \lambda - 2\gamma) \right] \nonumber \\
&&+ \; \Lambda(\lambda - \gamma) M_2^{(2)} ( \lambda, \lambda - \gamma, \lambda - 2\gamma ) \nonumber \\
&&+ \; \Lambda(\lambda - 2\gamma) M_1^{(2)} ( \lambda, \lambda - \gamma, \lambda - 2\gamma ) \; ,
\>
which simplifies to
\<
\label{rs3}
\Lambda(\lambda) \Lambda(\lambda - \gamma) \Lambda(\lambda - 2\gamma) &=& - \Lambda(\lambda) \prod_{j=1}^{L} \sinh{(\lambda - \mu_j)} \sinh{(\lambda - \mu_j - 2\gamma)} \nonumber \\
&& + \; \Lambda(\lambda - \gamma) 2 \cosh{(\gamma)} \prod_{j=1}^{L} \sinh{(\lambda - \mu_j)} \sinh{(\lambda - \mu_j - \gamma)} \nonumber \\
&& - \; \Lambda(\lambda - 2\gamma) \prod_{j=1}^{L} \sinh{(\lambda - \mu_j + \gamma)} \sinh{(\lambda - \mu_j - \gamma)}  \nonumber \\
\>
upon the use of (\ref{coeff}). The representation (\ref{wj}) can now be used in (\ref{rs3}).
In this way we set $\lambda = w_i + \gamma$ in (\ref{rs3}), and for this particular specialization 
we notice that the term $\left. \Lambda(\lambda - \gamma) \right|_{\lambda = w_i + \gamma}$ vanishes.
This procedure then yields the following equation determining the set of zeroes $\{ w_j \}$,
\<
\label{BAl3}
\prod_{k=1}^L \frac{\sinh{(w_i - \mu_k + \gamma)} \sinh{(w_i - \mu_k - \gamma)}}{\sinh{(w_i - \mu_k + 2\gamma)} \sinh{(w_i - \mu_k )}} = - \prod_{j=1}^{L-1} \frac{\sinh{(w_j - w_i + \gamma)}}{\sinh{(w_j - w_i - \gamma)}} \; .
\>
Although the equation (\ref{BAl3}) has been derived using the property (\ref{rou}) for the case $l=3$, we notice that (\ref{BAl3}) reduces to 
(\ref{BAl2}) for values of $\gamma$ belonging to $l=2$. Thus our results so far show that equation (\ref{BAl3}) is valid for
for both cases $l=2$ and $l=3$.

\bigskip 

\paragraph{$l=4$.} This particular root of unity condition also truncates the system of equations (\ref{FL}). For $l=4$ we are then left
with the relation
\<
\label{l4}
\left\{ \sum_{m=0}^{2} \; \sum_{0 \leq i_1 < \dots < i_{2m} \leq 3} V^{(2m)}_{i_{2m} , \dots, i_1} \; \widehat{\prod}_{\lambda}^{i_1, \dots , i_{2m}} \right\} \left. \Lambda(\lambda) \right|_{\lambda_j = \lambda - j \gamma}   = 0 \; ,
\>
where the form of the functions $V^{(2m)}_{i_{2m} , \dots, i_1}$ are given by (\ref{VV}). 
By making explicit use of (\ref{VV}) we find that (\ref{l4}) simplifies to
\<
\label{l4ex}
&& \Lambda(\lambda) \Lambda(\lambda - \gamma) \Lambda(\lambda - 2\gamma) \Lambda(\lambda - 3\gamma) = \nonumber \\
&& + \; \Lambda(\lambda - \gamma) \Lambda(\lambda - 2\gamma) \frac{\sinh{(3\gamma)}}{\sinh{(\gamma)}} \prod_{k=1}^{L} \sinh{(\lambda - \mu_k)} \sinh{(\lambda - \mu_k - 2\gamma)} \nonumber \\
&& - \; \Lambda(\lambda - 2\gamma) \Lambda(\lambda - 3\gamma) \prod_{k=1}^{L} \sinh{(\lambda - \mu_k + \gamma)} \sinh{(\lambda - \mu_k - \gamma)} \nonumber \\
&& - \; \Lambda(\lambda) \Lambda(\lambda - \gamma) \prod_{k=1}^{L} \sinh{(\lambda - \mu_k - \gamma)} \sinh{(\lambda - \mu_k - 3\gamma)} \nonumber \\
&& - \; \Lambda(\lambda ) \Lambda(\lambda - 3\gamma) \prod_{k=1}^{L} \sinh{(\lambda - \mu_k )} \sinh{(\lambda - \mu_k - 2\gamma)} + \mathcal{Q}(\lambda) \; , \nonumber \\
\>
where the function $\mathcal{Q}(\lambda)$ is given by
\<
\label{QQ}
\mathcal{Q}(\lambda) &=& \frac{\sinh{(3 \gamma)}}{\sinh{(\gamma)}} \prod_{k=1}^L \sinh{(\lambda - \mu_k)}^2 \sinh{(\lambda - \mu_k - 2\gamma)}^2 \nonumber \\
&& - \; \prod_{k=1}^L \sinh{(\lambda - \mu_k + \gamma)} \sinh{(\lambda - \mu_k - 3\gamma)} \sinh{(\lambda - \mu_k - \gamma)}^2 \nonumber \\
&& - \; 2 \cosh{(2\gamma)} \prod_{k=1}^L \sinh{(\lambda - \mu_k )} \sinh{(\lambda - \mu_k - 2\gamma)} \sinh{(\lambda - \mu_k - \gamma)}^2 \; . \nonumber \\
\>  
In what follows we shall describe how one can extract a set of equations determining the set of zeroes $\{ w_j \}$ from the functional relation
(\ref{l4ex}). For that we assume the representation (\ref{wj}) and set $\lambda = w_i + \gamma$ in (\ref{l4ex}). 
Under this specialization the term $\left. \Lambda(\lambda - \gamma) \right|_{\lambda = w_i + \gamma}$ vanishes
and we are left with a relation depending on the function $\mathcal{Q}(w_i + \gamma)$. 
Next we consider the specialization $\lambda = w_i + 2\gamma$ such that $\left. \Lambda(\lambda - 2\gamma) \right|_{\lambda = w_i + 2\gamma} = 0$.
This procedure yields two equations: one involving the function $\mathcal{Q}(w_i + \gamma)$ and another
depending on $\mathcal{Q}(w_i + 2\gamma)$. However, we can readily verify that $\mathcal{Q}(\lambda) = \mathcal{Q}(\lambda + \gamma)$
under the root of unity condition $l=4$. This property allows us to eliminate the functions $\mathcal{Q}$ from our equations leaving us
with the following relation,
\<
\label{BAl4}
\prod_{k=1}^L \frac{\sinh{(w_i - \mu_k + \gamma)} \sinh{(w_i - \mu_k - \gamma)}}{\sinh{(w_i - \mu_k + 2\gamma)} \sinh{(w_i - \mu_k )}} = - \frac{\Lambda(w_i - \gamma)}{\Lambda(w_i + \gamma)} \; .
\>
By substituting the representation (\ref{wj}) into (\ref{BAl4}) we immediately recognize equation (\ref{BAl3}). 
Thus our analysis so far shows that the set of equations (\ref{BAl3}) is valid
for $l=2,3,4$. In fact, Eq. (\ref{BAl3}) seems to be valid for arbitrary roots of unity as we shall discuss.

\bigskip

\paragraph{General case.} For arbitrary values of $l$ the property (\ref{rou}) truncates the system
of functional relations (\ref{FL}) and we only need to consider the following equation,
\<
\label{lgen}
\left\{ \sum_{m=0}^{[l]/2} \; \sum_{0 \leq i_1 < \dots < i_{2m} \leq l-1}  V^{(2m)}_{i_{2m} , \dots, i_1} \; \widehat{\prod}_{\lambda}^{i_1, \dots , i_{2m}} \right\} \left. \Lambda(\lambda) \right|_{\lambda_j = \lambda - j \gamma}  = 0 \; .
\>
Then we assume the representation (\ref{wj}) and consider the sequence of specializations $\lambda = w_i + p \gamma$
for $1 \leq p \leq p-2$. This procedure yields one equation at each level of specialization.
Similarly to the case $l=4$, the resulting system of equations can be manipulated in order to find compact equations
determining the roots $w_j$. Although this last step involves the use of non-trivial properties satisfied by the functions
$V^{(2m)}_{i_{2m} , \dots, i_1}$, the procedure above described holds in general and its implementation for particular
values of $l$ leads to the very same equation (\ref{BAl3}). A rigorous proof of (\ref{BAl3}) for arbitrary values of $l$ is
still missing but our analysis so far leads us to conjecture that (\ref{BAl3}) is valid for general roots of unity.

\section{Concluding remarks}
\label{sec:conclusion}

The main result of this work is the formula (\ref{Lgen}) which states a relation
between the six-vertex model with DWBC and the anti-periodic Heisenberg chain.
This relation is a direct consequence of the algebraic-functional approach introduced
in  \cite{Galleas_2008} and refined in the series of works \cite{Galleas_2010, Galleas_2011, Galleas_2012}.
The Yang-Baxter algebra is the main ingredient for the derivation of (\ref{Lgen}) which allows us to
establish a very non-trivial relation between the zeroes of certain quantities related to six-vertex models
with different boundary conditions. In one hand we have the zeroes of the partition function of the six-vertex
model with DWBC. On the other hand, we have the zeroes of the transfer matrix eigenvalues associated with
the six-vertex model with a non-diagonal boundary twist. The relation between those zeroes is then precised
in (\ref{BAeven}) and (\ref{BAodd}).

In this work we have also analyzed the cases where the six-vertex model anisotropy parameter satisfies 
a root of unity condition. In that case we have found a compact set of equations, namely (\ref{BAl3}),
characterizing the zeroes of the eigenvalues $\Lambda$.

Boundary conditions of domain wall type can be formulated for a variety of lattice integrable systems.
In particular, the so called $8V$-SOS model also admits domain wall boundary conditions and it has been studied
through this algebraic-functional approach in \cite{Galleas_2011, Galleas_2012}. The latter consists of an 
elliptic integrable system and one might wonder if there exists a twisted transfer matrix such that a relation
analogous to (\ref{Lgen}) holds. This problem has eluded us so far and its investigation would probably bring
further insights into the structure of integrable solid-on-solid models.

\section{Acknowledgements}
\label{sec:ack}
The author thanks G. Arutyunov for very useful discussions. This work is 
supported by the Netherlands Organization for Scientific Research (NWO) under the VICI grant 680-47-602
and by the ERC Advanced grant research programme No. 246974, {\it ``Supersymmetry: a window to non-perturbative physics"}.
The author also thanks the D-ITP consortium, a program of the Netherlands Organization for Scientific Research (NWO) funded by the Dutch
Ministry of Education, Culture and Science (OCW).

\appendix

\section{The function $V^{(2m)}_{i_{2m} , \dots, i_1}$}
\label{sec:proof}

In this appendix we aim to discuss the derivation of formulae (\ref{Lgen}) and (\ref{VV}).
The function $V^{(2m)}_{i_{2m} , \dots, i_1}$ given by (\ref{VV}) follows straightforwardly from
the functions $M_i^{(n)}$ and $N_{j,i}^{(n)}$ defined in (\ref{coeff}). We shall start by reviewing
the cases $L=2$ and $L=3$ already discussed in \Secref{sec:SOL}.

\paragraph{$L=2$.} The first equation of (\ref{LL2}) can be rewritten as
\[
Z(X^{0,1}) k_0 = \Lambda(\lambda_0) \Lambda(\lambda_1) - M_1^{(1)}(\vec{X}^{0,1}) \; ,
\] 
and we can compare its RHS with (\ref{ZB2}). In this way we require that $V_{1,0}^{(2)} = - M_1^{(1)}(\vec{X}^{0,1})$.
The explicit evaluation of (\ref{VV}) then yields the following expression
\[
V_{1,0}^{(2)} = \frac{c(\lambda_0 - \lambda_1)}{b(\lambda_0 - \lambda_1)} \prod_{k=1}^2 a(\lambda_0 - \mu_k) b(\lambda_1 - \mu_k) + \frac{c(\lambda_1 - \lambda_0)}{b(\lambda_1 - \lambda_0)} \prod_{k=1}^2 a(\lambda_1 - \mu_k) b(\lambda_0 - \mu_k) \; , \nonumber \\
\]
which corresponds to $- M_1^{(1)}(\vec{X}^{0,1})$ according to (\ref{coeff}).

\paragraph{$L=3$.} Similarly to the previous case, we firstly rewrite the first equation of (\ref{LL3}) as
\<
\label{T3}
Z(X^{0,2}) k_0 &=& \Lambda(\lambda_0) \Lambda(\lambda_1) \Lambda(\lambda_2) - \Lambda(\lambda_0) \left[ M_1^{(1)} (\vec{X}^{1,2}) + N_{2,1}^{(2)} (\vec{X}^{0,2}) \right] \nonumber \\
&& -  \; \Lambda(\lambda_1) M_2^{(2)} ( \vec{X}^{0,2} ) - \Lambda(\lambda_2) M_1^{(2)} ( \vec{X}^{0,2} ) \; . \nonumber \\
\>
We can now compare the RHS of (\ref{T3}) with (\ref{ZB3}). By doing so we find the following conditions:
\<
\label{cnd}
V_{1,0}^{(2)} &=& - M_1^{(2)} ( \vec{X}^{0,2} ) \nonumber \\
V_{2,0}^{(2)} &=& - M_2^{(2)} ( \vec{X}^{0,2} ) \nonumber \\
V_{2,1}^{(2)} &=& - M_1^{(1)} (\vec{X}^{1,2}) - N_{2,1}^{(2)} (\vec{X}^{0,2}) \; .
\>
It is now a straightforward computation to verify that the functions $V_{1,0}^{(2)}$, $V_{2,0}^{(2)}$ and $V_{2,1}^{(2)}$
obtained from (\ref{VV}) satisfy the conditions (\ref{cnd}) with functions $M_i^{(n)}$ and $N_{j,i}^{(n)}$ given by (\ref{coeff}). 
It is also worth remarking that the verification of the third condition of (\ref{cnd}) involves the simplification of functions
as it typically occurs in algebraic Bethe ansatz framework.

\paragraph{$L=4$.} We start from (\ref{FL}) and for the case $L=4$ we set $n=0,1,2,3,4$. In this way we are left with a total of $5$ equations
for the functions $\mathcal{F}_n$ which can be solved in favor of the eigenvalue $\Lambda$. The resulting equation will then
depend on the functions $M_i^{(n)}$ and $N_{j,i}^{(n)}$ defined in (\ref{coeff}), and it can be directly compared with (\ref{Lgen}).
By doing so we find the following conditions,
\begin{align}
\label{cnd1}
V_{3,0}^{(2)} = & - M_3^{(3)} (\vec{X}^{0,3}) & V_{3,2}^{(2)} & = - M_1^{(1)} (\vec{X}^{2,3}) - N_{2,1}^{(2)} (\vec{X}^{1,3}) - N_{3,2}^{(3)} (\vec{X}^{0,3}) \nonumber \\  
V_{2,0}^{(2)} = & - M_2^{(3)} (\vec{X}^{0,3}) & V_{3,1}^{(2)} & = - M_2^{(2)} (\vec{X}^{1,3}) - N_{3,1}^{(3)} (\vec{X}^{0,3}) \nonumber \\
V_{1,0}^{(2)} = & - M_1^{(3)} (\vec{X}^{0,3}) & V_{2,1}^{(2)} & = - M_1^{(2)} (\vec{X}^{1,3}) - N_{2,1}^{(3)} (\vec{X}^{0,3}) \; ,
\end{align} 
in addition to
\<
\label{cnd2}
V_{3,2,1,0}^{(4)} &=& M_1^{(3)}(\vec{X}^{0,3}) M_1^{(1)} (\vec{X}_{1}^{1,3}) + M_2^{(3)}(\vec{X}^{0,3}) M_1^{(1)} (\vec{X}_{2}^{1,3}) +  M_3^{(3)}(\vec{X}^{0,3}) M_1^{(1)} (\vec{X}_{3}^{1,3}) \nonumber \\
&& + N_{2,1}^{(3)}(\vec{X}^{0,3}) M_1^{(1)} (\vec{X}_{1,2}^{0,3}) + N_{3,1}^{(3)}(\vec{X}^{0,3}) M_1^{(1)} (\vec{X}_{0,3}^{0,3}) + N_{3,2}^{(3)}(\vec{X}^{0,3}) M_1^{(1)} (\vec{X}_{2,3}^{0,3}) \; . \nonumber \\
\>
Both relations (\ref{cnd1}) and (\ref{cnd2}) can be readily verified with the help of (\ref{VV}) and (\ref{coeff}). 

\paragraph{General $L$.} The structure of the function $V^{(2m)}_{i_{2m} , \dots, i_1}$ for arbitrary values of $L$ is obtained
from particular combinations of the functions $M_i^{(n)}$ and $N_{j,i}^{(n)}$. These combinations are built by eliminating
the functions $\mathcal{F}_n$ from the system of equations (\ref{FL}) in favor of the eigenvalue $\Lambda$. By carrying out this
procedure we find the relation (\ref{Lgen}). The function $V^{(2m)}_{i_{2m} , \dots, i_1}$, as defined in (\ref{VV}), captures
the aforementioned combinations of $M_i^{(n)}$ and $N_{j,i}^{(n)}$ which are explicitly given by (\ref{coeff}). Although it is a cumbersome 
computation, the simplifications required to arrive at formula (\ref{VV}) are performed in much the same spirit of
the algebraic Bethe ansatz.

\bibliographystyle{hunsrt}
\bibliography{references}

\end{document}